\documentclass[aps,pra,showpacs,twocolumn,superscriptaddress,bibnotes,floatfix,10pt]{revtex4-1}

\usepackage{mathtools}
\usepackage{graphicx}
\usepackage{float}
\usepackage[hidelinks]{hyperref}
\usepackage{dcolumn}
\usepackage{amsmath,amssymb,amsthm}
\usepackage{textcomp}
\usepackage{bm}
\usepackage{verbatim}
\usepackage{natbib}
\usepackage[normalem]{ulem}
\usepackage{color}
\usepackage{ulem}

\usepackage[american,british]{babel} 
\usepackage[utf8]{inputenc} 
\usepackage[T1]{fontenc}

\usepackage{hyperref}
\usepackage{xcolor}
\hypersetup{
	colorlinks   = true, 
	urlcolor     = blue, 
	linkcolor    = blue, 
	citecolor   = blue 
}

\newtheorem{definition}{Definition}

\newtheorem*{prop}{Proposition 2}
\newtheorem{proposition}{Proposition}

\DeclarePairedDelimiter\ket{\lvert}{\rangle}
\DeclarePairedDelimiterX\braket[2]{\langle}{\rangle}{#1 \delimsize\vert #2}
\newcommand{\id}{\mathbb{I}}
\newcommand{\sx}{\sigma^{x}}
\newcommand{\sz}{\sigma^{z}}
\newcommand{\sy}{\sigma^{y}}

\newcommand{\cd}{c^{\dagger}}
\newcommand{\g}{\gamma}
\newcommand{\gd}{\gamma^{\dagger}}

\newcommand{\tkh}{\theta_{k}/2}
\newcommand{\tkhh}{\theta_{k'}/2}

\newcommand{\ketbra}[2]{\left|#1\middle\rangle\middle\langle#2\right|}

\newcommand{\norm}[1]{\left\|#1\right\|}

\newcommand{\mean}[1]{\left\langle#1\right\rangle}

\newcommand{\UFSC}{Departamento de Física, Universidade Federal de Santa Catarina,
88040-900, Florianópolis, SC, Brazil}

\begin{document}
\title{Quantum Statistical Complexity Measure as a Signalling of Correlation Transitions} 
\author{André T. Ces\'{a}rio} 
\email{andretanus@gmail.com}
\affiliation{Departamento de F\'{\i}sica - ICEx - Universidade Federal de Minas Gerais, Av. Pres. Ant\^onio Carlos 6627 - Belo Horizonte - MG - Brazil - 31270-901.}
\author{Diego L. B. Ferreira}\affiliation{Departamento de F\'{\i}sica - ICEx - Universidade Federal de Minas Gerais,
Av. Pres. Ant\^onio Carlos 6627 - Belo Horizonte - MG - Brazil - 31270-901.}
\author{Tiago Debarba}
\email{debarba@utfpr.edu.br}
\affiliation{Departamento Acad{\^ e}mico de Ci{\^ e}ncias da Natureza, Universidade Tecnol{\'o}gica Federal do Paran{\'a} (UTFPR), Campus Corn{\'e}lio Proc{\'o}pio, Avenida Alberto Carazzai 1640, Corn{\'e}lio Proc{\'o}pio, Paran{\'a} 86300-000, Brazil.}
\author{Fernando Iemini}\affiliation{Instituto de Física, Universidade Federal Fluminense, 24210-346 Niter\'oi, Brazil.}
\author{Thiago O. Maciel}
\affiliation{\UFSC}
\author{Reinaldo O. Vianna}\affiliation{Departamento de F\'{\i}sica - ICEx - Universidade Federal de Minas Gerais,
Av. Pres. Ant\^onio Carlos 6627 - Belo Horizonte - MG - Brazil - 31270-901.}
\date{\today}

\begin{abstract}
We introduce a quantum version for the statistical complexity measure, in the context of quantum information theory, and use it as a signalling function of quantum order-disorder transitions.  We discuss the possibility for such transitions to characterize interesting physical phenomena, as quantum phase transitions, or abrupt variations in the  correlation distributions. We apply our measure to two exactly solvable Hamiltonian models, namely: the $1D$-Quantum Ising Model and the Heisenberg XXZ spin-$1/2$ chain. 
We also compute this measure for one-qubit and two-qubit reduced states for the considered models, and analyse its behaviour across its quantum phase transitions for finite system sizes as well as in the thermodynamic limit by using Bethe ansatz.
\end{abstract}
\maketitle

\section{Introduction}
\label{intro}
Let us consider a physical, chemical or biological process, which can be for example the change in the temperature of the water, the mixture between two solutions or the formation of a neural network. It is intuitive to believe that if one can classify all the possible configurations of such systems, described by their \emph{ordering and disordering patterns}, it would be possible to characterize and control them. As, for example, during the process of changing the temperature of water, by knowing the pattern of ordering and disordering of its molecular structure, it would be possible to characterize and control completely its phase transitions.  

In information theory, the ability to identify certain patterns of order and disorder of probability distributions enables to control the creation, transmission and measurement of information. In this way, 
the characterization and quantification of \emph{complexity} contained in physical systems and their constituent parts is a crucial goal for information 
theory  \cite{badii1997}. One point of consensus in the literature about complexity is that no formal definition of this term exists. Intuition suggests that systems which can be described as “not complex” are readily comprehended: they can be described concisely by means of few parameters or variables, and their information content is low. There are considerable ways to define measures of the degree of complexity of physical systems. Among such definitions, we can mention measures based on data compression algorithms of finite size sequences \cite{lempel76,montano02,Szczepanski09}, Komolgorov or Chaitin measures based on the size of the smallest algorithm that can reproduce a particular type of pattern \cite{kolmogorov65, chaitin66}, and measures concerning the classical information theory \cite{martin03,lamberti04,binder00,shiner99,toranzo14,wackerbauer94,zurek90}. 
Such quantifiers must satisfy some properties: (\textit{i.}) assign a minimum 
 value (possibly zero) for opposite extremes of order 
 and disorder; (\textit{ii.}) should be sensitive to transitions of order-disorder patterns; (\textit{iii.}) and must be computable.   

The statistical complexity assigns the simplicity of a probability distribution to the amount of resources needed to store information \cite{Crutchfield09,Riechers16}. Similarly, in the quantum realm, the complexity of a given density matrix
could be translated as the resource needed to create, operate or measure the quantum state of the system \cite{Gu12,Yang18,Thompson18}. On the other hand, the quantum information meaning of complexity could play an important role in the quantification of transitions of order and disorder patterns, which could indicate some quantum physical phenomenon, as for example quantum phase transitions. 

Regarding the complexity contained in systems, some of the simplest models in physics are:  the ideal gas and the perfect crystal. In an ideal gas model, the system can be found with the same probability in any of the available micro-states, therefore each state contributes equally to the same amount of information. On the other hand, in a 
perfect crystal, the symmetry rules restrict the accessible states of the system to only one very symmetric state. These simple models are extreme cases of minimum complexity, in a scale of order and disorder, therefore, there might exist some intermediate state which contains a maximum complexity value in that scale \cite{lopezruiz95}.

The main goal of this work is to introduce a quantum version of the statistical complexity measure, based on the physical meaning of the characterization and quantification of transitions between order-disorder patterns of quantum systems.  As by-product, to apply this measure in the study of quantum phase transitions. Physical properties of systems across a quantum phase transition are dramatically altered, and in this way it is interesting to understand how the complexity of the system would behave under such transitions. In our analysis, {we} studied the parametric Ising model and the Heisenberg XXZ-$1/2$ model.
 
 The manuscript is organized as follows: In Sec. \ref{complex.class}, we introduce the Statistical Measure of Complexity defined by L\'{o}pes-Ruiz \emph{et al.}, and in Sec. \ref{complex.quant} we introduce the quantum counterpart of this measure: the quantum statistical measure of complexity. We present some properties of this measure and we exhibit a closed expression of this measure for one-qubit. In Sec. \ref{examples} we discuss two interesting examples and applications: the $1D$ quantum Ising model (Sec. \ref{ex.ising}), in which we compute the quantum statistical measure of complexity for one-qubit reduced state from $N$ spins, in the thermodynamic limit with the objective of determining the quantum phase transition point.  We further determine the first-order quantum transition point and the continuous quantum phase transition for the Heisenberg XXZ spin-$1/2$ model, with $h=0$ (Sec. \ref{ex.xxz}), by means of the quantum measure of statistical complexity of the two-qubit reduced state of the nearest neighbours, in the thermodynamic limit. Finally, we give concluding remarks in Sec. \ref{conclusao}.
 
\section{Classical Statistical Complexity Measure - CSCM}
\label{complex.class}
Consider a system possessing $N$ accessible states $\{x_1,x_2,\cdots,x_N\}$, when observed on a particular scale, with each state having an intrinsic probability given by $\vec{p} = \{p_i\}_{i=1}^N$. As discussed before, the candidate function to quantify the complexity of a probability distribution associated with a physical system, must attribute zero value for systems with the maximum degree of order, that is, for pure distributions: $\vec{p} = \{p_i = 1, p_{j\neq i}=0\}$, and also assign zero for disordered system which are characterized by identically distributed vectors (i.i.d.): $\vec{\mathcal{I}} =\{\mathcal{I}_i = 1/N\}$, for all $i=1,\ldots, N$. Let us address the case of ordered and disordered systems separately.

\subsection{Degree of Order}
A physical system possessing the maximum degree of order can be regarded as a system with a symmetry of all of its elements. The probability distribution that describes such systems is best represented by a pure vector, which places the system as having only one possible configuration. Physically this is the case of a gas at zero Kelvin temperature, or a perfect crystal where the symmetry rules restrict the accessible state to a very symmetric one. In order to quantify the degree of order of a given system, the function must assign maximum value for pure probability distributions, and attribute zero for equiprobable configurations. A function capable of quantifying such degree of order is the $l_1$-distance between the probability distribution and the identically distributed vector (i.i.d.):
\begin{equation}
\label{tracedist}
D(\vec{p},\vec{\mathcal{I}}) = \frac{1}{2}||\vec{p} - \vec{\mathcal{I}}||_1 = \frac{1}{2}\sum_i\sqrt{ \left(p_i - \frac{1}{N}\right)^2}, 
\end{equation}
where $\vec{\mathcal{I}}$ is a vector with elements $\mathcal{I}_i = 1/N$, for all $i=1,\ldots, N$. This function plays a role of a \emph{disequilibrium} function and it quantifies the order of a probability vector.
It consists on the sum of the absolute values of the elements of the vector $\vec{p} - \vec{\mathcal{I}}$. The measure $D$ will have zero value for maximally disordered systems and maximum value for maximally ordered systems. 

\subsection{Degree of Disorder}
In contrast with an ordered system, a system possessing the maximum degree of disorder is described by an equiprobable distribution.  This means an equally probable expectation of occurring any of its configurations, as in a fair dice game, or in a partition function of an isolated ideal gas. From a statistical point of view, the probability vector that describes this feature is the identically distributed vector (i.i.d.) $\vec{\mathcal{I}}$, as defined above. One can define the degree of disorder of a system as a function which assigns value of zero for pure probabilities distributions, (associated with maximally ordered distributions), and a maximum value for i.i.d. distributions. A well known function capable of quantifying the degree of disorder of a probability vector is the Shannon entropy: 
\begin{equation}
\label{Shannon}
H(\vec{p}) = -\sum_{i=1}^N p_i \log(p_i).
\end{equation}
In this way, the Shannon entropy $H(\vec{p})$ will assign zero for maximally ordered systems, and a maximal value for i.i.d vectors equals to $\log N$. The $\log$ function is taken in basis $2$ in order to quantify the amount of disorder in bits. 

\subsection{Quantifying Classical Complexity}
With all these extreme behaviors of order: (Eq. \eqref{tracedist}) and disorder: (Eq. \eqref{Shannon}) in mind, L\'{o}pes-Ruiz \emph{et al.} defined a classical statistical measure of complexity constructed as a product of such order-disorder quantifiers. This function should deal with the intermediate states of order and disorder, which can be associated with a kind of a complex behavior, measuring the amount of complexity of a statistical physical system \cite{lopezruiz95}.
\begin{definition}[Classical Statistical Measure of Complexity]
Let us consider a probability vector given by $\vec{p} = \{p_i\}_{i=1}^N$, with $dim(\vec{p}) = N$, associated with a random variable $X$, representing all possible states of a system. The function $\mathcal{C}(\vec{p})$ is a measure of the system's complexity and can be defined as: 
\begin{equation}
\mathcal{C}(\vec{p}) =\frac{1}{\log{N}} H(\vec{p})D(\vec{p},\vec{\mathcal{I}}).
\label{complex.classica}   
\end{equation}
The function $\mathcal{C}(\vec{p})$ will vanish for simple systems, such as the ideal gas model or a crystal, and it should reach a maximum value for some state.
\end{definition}
The classical statistical measure of complexity depends on the nature of the description associated to a system and with the scale of observation \cite{lopezruiz95}. This function, generalized as a functional of a probability distribution, has a relation with a time series generated by a classical dynamical system \cite{rosso13}. Two ingredients are fundamental in order to define such a quantity: the first one is an entropy function which quantifies the information contained in a system, and could also be the Tsallis' Entropy \cite{tsallis}, Escort-Tsallis \cite{escort.tsallis} or Rényi Entropy \cite{renyi}. The other ingredient is the definition of a distance function in the state of probabilities, which indicates the disequilibrium relative to a fixed distribution (in this case the distance to the i.i.d. vector). For this purpose we can use an Euclidean Distance (or some other $p$-norm), the Bhattacharyya Distance \cite{bhatt} or Wooters' Distance \cite{wooters}. We can also apply a statistical measure of divergence, for example the Classical Relative Entropy \cite{kl}, Hellinger distance and also Jensen-Shannon Divergence \cite{HSJS}. We make note of some other generalized versions of complexity measures in recent years, and these functions have proved to be useful in some branches of classical information theory \cite{martin03,lopezruiz09,sanudo08,montgomery08,sen11,sanudo09,moustakids10,moreno14,calbet01,ruizSanudo12}.

\section{Quantum Statistical Complexity Measure - QSCM}
\label{complex.quant}
\subsection{Quantifying Quantum Complexity}
The quantum version of the statistical complexity measure quantifies the amount of order-disorder complexity in a quantum system. For quantum systems, the probability distribution is replaced by a density matrix (positive semi-definite and trace one). Likewise the classical case, the extreme cases of order and disorder are respectively the pure quantum states $\ketbra{\psi}{\psi}$, and the maximally mixed state: $\mathcal{I} := \mathbb{I}/N$, where $N$ is the dimension of these two quantum systems. Also in analogy with the description for classical probability distributions, the quantifier of quantum statistical complexity must be zero for maximum degree of order and disorder. One can define the Quantum Statistical Complexity Measure (QSCM) as a product of an order and a disorder quantifiers: a quantum entropy, which measure the amount of disorder related to a quantum system and a pairwise distinguishability measure of quantum states which plays the role of a \emph{disequilibrium} function. One of the functions to measure the amount of disorder of a quantum system is the von Neumann entropy and it is given by:
\begin{equation}
S(\rho) = -\text{Tr}[\rho \log (\rho)],
\end{equation}
where $\rho$ is the density matrix of the system. The trace distance between $\rho$ and the the maximally mixed state quantifies the order:
\begin{equation}
D(\rho,\mathcal{I}) \equiv \norm{\rho - \mathcal{I}}_1 = \frac{1}{2}\text{Tr}\sqrt{\left(\rho - \mathcal{I}\right)^2}.
\end{equation}

\begin{definition}[Quantum Statistical Complexity Measure (QSCM)]\label{d2}
Let $\rho\in\mathcal{D}(\mathcal{H}_N)$ be a quantum state over an $N-$dimensional Hilbert space. Then we can define a quantifier of quantum statistical measure of complexity as the following functional of $\rho$: 
\begin{equation}
\mathcal{C}(\rho) = \frac{1}{\log{N}} S(\rho)\cdot D(\rho,\mathcal{I}),
\label{complexquantica}   
\end{equation}
where  $S(\rho)$ is the von Neumann entropy, and  $D(\rho,\mathcal{I})$ is a distinguishability quantity between the state $\rho$ and the normalized maximally mixed state $\mathcal{I}$.   
\end{definition}

In analogy with the classical counterpart, in the definition of quantum statistical complexity measure, there is a \emph{carte blanche} in choosing the quantum entropy function, as for example the quantum Rényi entropy \cite{qrentropy}, or quantum Tsallis entropy \cite{petz}. Similarly, we can choose other disequilibrium function as a measure of distinguishability of quantum states. It can be some Shatten-$p$ norm \cite{bhatia}, or a quantum Rényi relative entropy \cite{rajagopal}, the quantum skew divergence \cite{audeanert14}, or a quantum relative entropy \cite{schumacher00}. Another feature that might generalize the quantity is to define a more general quantum state $\rho^*$ as a reference state, (rather than normalized maximally mixed state $\mathcal{I}$), in the disequilibrium function. This choice must be guided by some physical symmetry or interest. Some obvious candidates are the thermal mixed quantum state, and the canonical thermal pure quantum state \cite{sugiura13}. 

For our purposes here, we define the QSCM by means of the trace distance, between the quantum state $\rho$ and $\mathcal{I}$ acting as the disequilibrium function, once that it is the most distinguishable distance in Hilbert space, and also monotonic under stochastic operations.

\subsection{Some Properties of the QSCM}
\label{prop}
To complete our introduction of the quantifier of quantum statistical complexity, we should require some properties to guarantee a \emph{bona fide} information quantifier. The amount of order-disorder, as measured by the QSCM, must be invariant under unitary operations because it is related with the purity of quantum states.
\begin{proposition}[Unitary Invariance]
The Quantum Statistical Complexity Measure is invariant under unitary transformations:
\begin{equation}
\mathcal{C}(U \rho U^{\dagger}) = \mathcal{C}(\rho), 
\end{equation}
where $\rho\in\mathcal{D}(\mathcal{H}_N)$ and $U$ is a unitary transformation on $\mathcal{H}_N$. 
\end{proposition}
This statement comes directly from the invariance under unitary transformation of von Neumann entropy and trace distance. Another important property regards the case of inserting copies of the system in some experimental context. Let us consider an experiment in which the experimentalist must quantify the QSCM of a given state $\rho$ by means of a certain number $n$ of copies $\rho^{\otimes n}$, which therefore implies that the QSCM of the copies should be bounded by the quantity of only one copy. 
\begin{proposition}[Sub-additivity over copies]
\label{subad}
Given a product state $\rho^{\otimes n}$, with $dim(\rho) = N$, the QSCM is a sub-additive function over copies:
 \begin{equation}
\mathcal{C}(\rho^{\otimes n})\leq n \mathcal{C}(\rho). 
\end{equation}
\end{proposition}
Indeed this is an expected property for a measure of information, since the 
regularized number of bits of information gained from a given system 
cannot increase just by considering more copies of the same system. 
The proof of Proposition \ref{subad} is in Appendix \hyperref[appendixA]{A}, and it comes from the additivity of von Neumann entropy and sub-additivity of trace 
distance. 

It is important to notice that quantum statistical complexity is not sub-additive over general extensions with quantum states, for example, {\bf \emph{i.}} {\it extensions with maximally mixed states}: let us consider a given state $\rho$ is extended with one maximally mixed state $\mathcal{I}= \mathbb{I}/N$, with $\text{dim}(\rho) = \text{dim}(\mathcal{I}) = N$. 
\begin{align}
\mathcal{C}(\rho\otimes\mathcal{I}) 
&= \frac{S(\rho) + \log N}{2\log N}D(\rho,\mathcal{I}), 
\label{extension.rho.iN} \\
& = \frac{\mathcal{C}(\rho)}{2}+\frac{D(\rho,\mathcal{I})}{2}\\
& \leq\frac{\mathcal{C}(\rho)}{2} +\frac{1}{2}. 
\label{bound.ext.rho.iN.}
\end{align} 
Eq. \eqref{bound.ext.rho.iN.} presents an upper bound to the QSCM for this extended state. This feature displays that the measure of the compound state is bounded by the quantity of one copy. {\bf \emph{ii.}} In Eq. \eqref{d2.extension.rho.iN} we present the QSCM for a {\it more general extension given by} $\rho^{\otimes n}\otimes\mathcal{I}^{\otimes n}$. This feature also shows that the measure of the compound state is also bounded by the quantity of one copy. 
\begin{align}
\mathcal{C}(\rho^{\otimes n}\otimes\mathcal{I}^{\otimes n})&\leq\displaystyle  n\frac{S(\rho) + \log N}{2\log N}D(\rho,\mathcal{I}), \\
&\leq\displaystyle n\left(\frac{\mathcal{C}(\rho)}{2}+\frac{D(\rho,\mathcal{I})}{2}\right),\\
&\leq\displaystyle n\left(\frac{\mathcal{C}(\rho)}{2} +\frac{1}{2}\right).   
\label{d2.extension.rho.iN}
\end{align} 
{\bf \emph{iii.}} As a last example of {\it nonextensivity over general compound states}, let us consider the extension with a pure state $\ketbra{\psi}{\psi}$, with $\dim(\rho) = dim(\ketbra{\psi}{\psi})$. 
\begin{align}
\mathcal{C}(\rho\otimes\ketbra{\psi}{\psi}) &= \displaystyle\frac{S(\rho)}{2\log N}D(\rho,\ketbra{\psi}{\psi}),\\
&\leq\displaystyle\frac{S(\rho)}{2\log N}\left(D(\rho,\mathcal{I}) + \frac{N-1}{N}\right).
\label{extension.rho.pure}
\end{align}

As discussed above, the QSCM is a measure that intends to
detect changes on properties, as for example changes on 
patterns of order and disorder. Therefore, the measure
must be a continuous function over the parameters of the states responsible for its transitional
characteristics. Naturally, the quantum 
complexity is a continuous function, since it comes from the product of two continuous functions. Due to continuity, it is possible to define the derivative function of the quantum statistical complexity measure:
\begin{definition}[Derivative]
\label{deriv12}
Let us consider a physical system described by the one-parameter set of states:   
$\rho(\alpha) \in \mathcal{D}(\mathcal{H}_N)$, for $\alpha \in 
\mathbb{R}$. We can define the derivative with respect to $\alpha $ as:
\begin{equation}
\label{l1}
\frac{d\mathcal{C}}{d\alpha}:=\lim_{\varepsilon\to 
0}\frac{\mathcal{C}(\rho(\alpha+\varepsilon))-\mathcal{C}(\rho(\alpha))
}{\varepsilon}.
\end{equation}
\end{definition}
In the same way as defined in Def. \ref{deriv12}, it is possible to obtain higher order derivatives. 

The description of physical systems depends on measurable quantities such as temperature, interaction strength, interaction range, orientation of an external field, \emph{etc.} These quantities can be described by parameters in a suitable space, for example, let us consider a parameter describing some physical quantity $\alpha$, and a set of one parameter states $\rho(\alpha)$. The study of transitions between patterns of order and disorder with the objective of inferring physical properties of a system can generate great interest. 

At low temperatures, physical systems are typically ordered, increasing the temperature of the system, they can undergo phase transitions or order-disorder transitions into less ordered states: solids lose their crystalline form in a solid-liquid transition; iron loses magnetic order if heated above the Curie point in a ferromagnetic-paramagnetic transition, \emph{etc}. For many-particle and composed systems, local change on the order-disorder degree can be also associated to a transition in the correlations pattern. In this way, a detectable change in these parameters may indicate an alteration in system configuration which is considered here as a changing in the pattern of order-disorder, or in other words, as a correlation transition, as presented in Def. \ref{ordem.desordem}. 
\begin{definition}[Correlation Transition]
In many-particle systems a transition of correlations occurs when a system changes from a state that has a certain order, pattern, or correlation, to another state possessing another order or correlation.
\label{ordem.desordem}
\end{definition}
In an abstract manner, a quantum state undergoing a path through the i.i.d. identity matrix, is an example of such transitions which may have physical meaning as we will observe later in some examples. Let us suppose that a certain subspace of a quantum system can be interpreted as having a certain order and there exists a path in which it passes through the identity. This path can be analyzed as having an order-disorder transition. In order to illustrate the formalism of quantum statistical complexity in this context of order-disorder transition, In Sec. \ref{examples} we apply it to two well-known quantum systems that exhibit quantum phase transitions: the $1D$ quantum Ising Model (Sec. \ref{ex.ising}), and the Heisenberg XXZ spin-$1/2$ model chain (Sec. \ref{ex.xxz}).

\section{Examples and Applications}
\label{examples}
In this section we calculate analytically the expression of QSCM for qubits in Sec.~\ref{c.one}, and present the application of QSCM on physical system in orther to evince quantum phase transition and correlation ordering transitions for 1D-Quantum Ising Model in Sec.~\ref{ex.ising} and XXZ-Heisenberg $1/2$-spin chain in Sec.~\ref{ex.xxz}. 
\subsection{QSCM of One-Qubit}
\label{c.one}

Let us suppose we have a one-qubit state $\rho$, written in the Bloch basis. In Eq. \eqref{complexidade.um}, we exhibit analytically the Quantum Statistical Complexity Measure $\mathcal{C}(\rho)$ of one-qubit, written as $\rho = \frac{1}{2}(\mathbb{I}+\vec{r}\cdot\vec{\sigma})$:
\begin{equation}
\mathcal{C}(r) = -\frac{r^2 }{2}\mathrm{arctanh}(r)-\frac{r}{4}\log \left(\frac{1-r^2}{4}\right).
\label{complexidade.um} 
\end{equation}
Where $\vec{r} = (x,y,z)$, $r = |\vec{r}| = \sqrt{x^2 + y^2 + z^2}$, $0\leq r \leq 1$, and $\vec{\sigma}$ is the Pauli matrix vector.

It is interesting to notice that the quantum statistical complexity of one-qubit written in the Bloch basis is a function dependent only on $r$. This expression will be useful in the study of quantum phase transitions for example in the $1D$-Ising model, discussed in Sec. \ref{ex.ising}, where an analytical expression for the state of one-qubit reduced from $N$ spins, in the thermodynamic limit, will be obtained. Other useful expressions can be obtained, for example, the trace distance between the state and the normalized identity for one-qubit is also a function of $r$, in the Bloch's basis, $D(r) = r/2$, and therefore, the entropy function can be easily written as $S(r) = 2C(r)/r$ by using Eq. \eqref{complexidade.um}. In addition we exhibit the first (Eq. \eqref{dC}), and second (Eq. \eqref{d2C}) derivatives of QSCM, for one-qubit written in the Bloch Basis. One can observe that these functions also depend only on $r$: 
\begin{equation}
\frac{d\mathcal{C}(r)}{dr} = -r\cdot\text{arctanh}(r)-\frac{1}{4} \text{log}\left(\frac{1-r^2}{4}\right),
\label{dC}
\end{equation}
\begin{equation}
\frac{d^2\mathcal{C}(r)}{dr^2} = \frac{r}{2 \left(1-r^2\right)}-\text{arctanh}(r).
\label{d2C} 
\end{equation}

\subsection{\texorpdfstring{$1D$}{1D} Quantum Ising Model}
\label{ex.ising}

The $1D$ quantum Ising Model presents a quantum phase transition and, despite its simplicity, still
generates a lot of interest from the research community. One of the motivations lies in the fact that spin chains possess a great importance in modelling quantum computers. The Hamiltonian of the quantum Ising model is given by:
\begin{equation}\label{ising}
\mathcal{H} = -J\displaystyle\sum_{j=1}^{N}\sigma_{j}^{x}\sigma_{j+1}^{x} - g\sigma_j^{z},
 \end{equation}
where $\{\sigma^x, \sigma^y, \sigma^z \}$ are the Pauli matrices, $J$ is an exchange constant that sets the interaction 
strength between the pairs of first neighbors $\{j,j+1\}$, and $g$ is a parameter that 
represents an external transverse field. Without loss of generality, we can set $J=1$, since it simply defines an energy scale for the Hamiltonian. The Ising model ground state can be obtained analytically by a diagonalization consisting of three steps: 
\begin{itemize}
    \item A Jordan-Wigner transformation: 
    \[ \sz_{j} \longrightarrow 1-2c_{j}^{\dagger}c_{j}, \]
    where $c_j$ and $c^{\dagger}_{j}$ are the annihilation-creation operators, respecting the anti-commutation relations: $\{c_{j},\cd_{k}\} =\delta_{jk}I$ and $\{c_{j},c_{k}\}=0$;
    \item a Discrete Fourier Transform (DFT): 
    \[c_{j} \longrightarrow \displaystyle\frac{1}{\sqrt{N}}\sum_{k=0}^{N-1}c_{k}\,e^{2\pi i(kj)/N}, \]
    \item a Bogoliubov transformation:
    \[c_{k} \longrightarrow \cos(\theta_k/2)\g_{k}-\sin(\theta_k/2)\gd_{-k}, \]
\end{itemize}
where $\theta_k$ represents the basis rotation from the mode $c_k$ to the new mode representation $\g_k$.
The angles $\theta_{k}$ are chosen such that the ground state of the Hamiltonian in Eq.\eqref{ising} is the vacuum state in $\gamma_{k}$ mode representation,  and it is given by $\theta_k = \mathrm{arctan}\left(\frac{\sin(k)}{g-\cos(k)}\right)$  \cite{nielsen02}. 

We can calculate the reduced density matrix of one spin by using the Bloch representation, in which all coefficients are obtained \emph{via} expectation values of Pauli operators. The one-qubit state in the site $j$, $\rho_{j}^{(1)}$ can be written as 
\begin{equation}
 \rho_{j}^{(1)} = \frac{\id}{2} + \frac{\vec{r}_{j}\cdot\vec{\sigma}_{j}}{2},
\end{equation}
where $r^{a}_{j}=\left<\sigma^{a}_{j}\right>$  are expected values in vacuum state in the site $j$, and $a=x,y,z$. 
Note that $\sx_{j},\sy_{j}=0\quad \forall j$, because they combine an odd number of fermions. Therefore, the Bloch's vector possesses only the $z$-component. Let us define $\tkh = \beta_k$, and $\tkhh = \beta_{k'}$. Thus, the $z$-component will be given by
\begin{equation*}
\begin{aligned}
 \sz_{j} &= 1-2c_{j}^{\dagger}c_{j},\\
         &= 1-\frac{2}{N}\sum_{k,k'}e^{-i(k-k')j}c_{k}^{\dagger}c_{k'},\\
         &= 1-\frac{2}{N}\sum_{k,k'}e^{-i(k-k')j}\left(\cos(\beta_k)\gd_{k}-\sin(\beta_k)\g_{-k}\right)\\
         &\times\left(\cos(\beta_{k'})\g_{k'}-\sin(\beta_{k'})\gd_{-k'}\right).\\
         &= 1-\frac{2}{N}\sum_{k,k'}e^{-i(k-k')j}\left(\cos(\beta_k)\cos(\beta_{k'})\gd_{k}\g_{k'}-\right.\\
         &\left.\cos(\beta_k)\sin(\beta_{k'})\gd_{k}\gd_{-k'}-\sin(\beta_k)\cos(\beta_{k'})\g_{-k}\g_{k'}\right.\\
         & \left.+\sin(\beta_k)\sin(\beta_{k'})\g_{-k}\gd_{-k'}\right).
\end{aligned}
\end{equation*}

As discussed above, the only non-vanishing term will be $\left<\sz_{j}\right>$, and therefore:

\begin{align}
 \left<\sz_{j}\right>
         &= 1-\frac{2}{N}\sum_{k,k'}e^{-i(k-k')j}\left(\cos(\beta_k)\cos(\beta_{k'})\left<\gd_{k}\g_{k'}\right> - \right.\nonumber\\
         &\left.\cos(\beta_k)\sin(\beta_{k'})\left<\gd_{k}\gd_{-k'}\right> -\sin(\beta_k)\cos(\beta_{k'})\left<\g_{-k}\g_{k'}\right> \right.\nonumber\\
         &\left.+ \sin(\beta_k)\sin(\beta_{k'})\left<\g_{-k}\gd_{-k'}\right>\right),\nonumber\\
         &= 1-\frac{2}{N}\sum_{k,k'}e^{-i(k-k')j}\sin(\beta_k)\sin(\beta_{k'})\delta_{k,k'},\nonumber\\
\left<\sz_{j}\right> &= 1-\frac{2}{N}\sum_{k}\sin^{2}(\tkh).\label{bloch.rho1}
\end{align}
In Eq.\eqref{diego} we exhibit the one-qubit reduced density matrix  in Bloch's representation ($\rho_1$):
\begin{equation}
\rho_1 = \frac{\mathbb{I}}{2}+\left( \frac{1}{2}-\frac{1}{N}\sum\limits_{k\in\mathcal{K}}\sin^2\left(\frac{\theta_k}{2}\right)\right)\sigma^z,
\label{diego}   
\end{equation}
where $\theta_k$ is the Bogoliubov rotation angle and the summation index $k\in\mathcal{K}$, with $\mathcal{K} = [\pm\frac{\pi}{N},\pm\frac{3\pi}{N},\cdots,\pm\left(\pi - \frac{2\pi}{N}\right)]$. This result is independent of the spin index, as expected for systems that are translational invariant.

We can now calculate QSCM for the reduced density matrix analytically. From Eq.~\eqref{complexidade.um}, we simply identify the Bloch vector of the reduced density matrix having only $z$-component, as written in Eq. \eqref{bloch.rho1}. This quantity in the thermodynamic limit can be obtained by taking the limit of the Riemann sums, $\mean{\sz}(g) = \text{lim}_{N\to\infty}\sum_{k=1}^N \mean{\sz_{k}}$. This Bloch's vector component is a function of the field $g$, that is, $\mean{\sz}(g) = \text{lim}_{N\to\infty}\sum_{k=1}^N 1-\frac{2}{N}\sum_{k}\sin^{2}(\theta(k))$, with $\theta(k) = \mathrm{arctan}\left(\frac{\sin(l)}{g-\cos(l)}\right)$, and $l = \frac{(2k-1)\pi}{N} - \pi$. Thus, the $z$ component of Bloch's vector given in Eq. \eqref{bloch.rho1} goes to the following integral, written as: 
\begin{equation}
\mean{\sz}(g) = 1 - \frac{2}{\pi}\int\limits_{-\pi}^{0} \sin^2\left( \frac{1}{2}\text{arctan}\left(\frac{\text{sin}(\xi)}{g - \text{cos}(\xi)}\right)\right)d\xi.
\label{bloch.integral}
\end{equation}
This integral can be solved analytically in the thermodynamic limit for some values of the transverse field parameter.
For $g=0$, we can easily obtain $\mean{\sigma_z} = 0$, which corresponds to a one-qubit maximally mixed reduced state $\rho_1 = \mathbb{I}/2$. At $g=1$, \emph{i.e.}, in the critical point, the integral given in Eq. \eqref{bloch.integral} can be also solved and we obtain $\mean{\sigma_z} = 2/\pi$, in the thermodynamic limit. The eigenvalues of the one-qubit reduced state $\rho_1$, at $g=1$, can be obtained analytically as: $\{ 1/2 \pm 1/\pi \}$. 
 For other values of $g$, the integral written in Eq. \eqref{bloch.integral} can be written as elliptic integrals of first and second kinds \cite{pfeuty70}. By using the result given in Eq. \eqref{bloch.integral} on Eq. \eqref{complexidade.um}, we can thus obtain the quantum statistical complexity measure for one-qubit reduced density matrix in the thermodynamic limit as a function of the transverse field parameter $g$. 

In Fig. \ref{D2complex.rho3} we present the second derivative of QSCM measure with respect to the transverse field parameter $g$, for different finite system sizes $N=4,8,16,1000$. We Also calculated this derivative in the thermodynamic limit by using Eq. \eqref{d2C} and Eq. \eqref{bloch.integral}.

\begin{figure}[h]
\includegraphics[width=0.9\columnwidth]{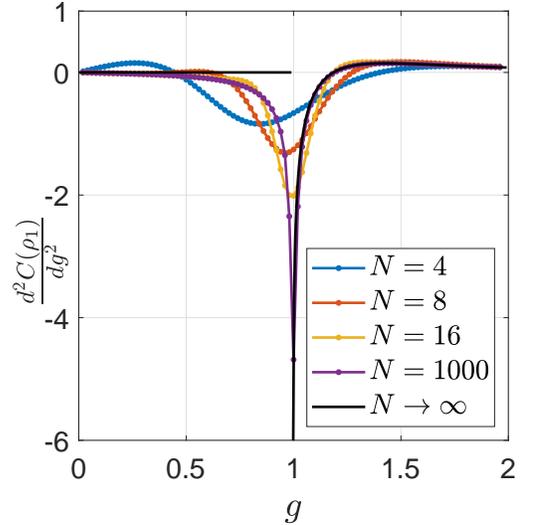}
\caption{Second derivative of QSCM with respect to the transverse field parameter $g$, for different finite system sizes: $N = 4,6,16,1000$ and for the thermodynamic limit (continuous line), $N\to\infty$, for $g\in [0,2]$.}
\label{D2complex.rho3}
\end{figure}
It is well known that at $g=1$, there is a quantum phase transition of second order \cite{damski}. By observing Fig. \ref{D2complex.rho3}, we can directly recognize a sharp behaviour of the measure in the transition point $g=1$.

\subsection{XXZ-\textonehalf\hspace{2.0mm}Model} 
\label{ex.xxz}
Quantum spin models as XXZ-$1/2$ model can be simulated experimentally by using Rydberg-excited atomic ensembles in magnetic microtrap arrays \cite{whitlock17}, and also by a low-temperature scanning tunnelling microscopy \cite{toskovic16} among many others quantum simulation experimental arrangements. Let us consider a Heisenberg XXZ spin-$1/2$ model defined by the following Hamiltonian: 
\[
\mathcal{H} = -J\displaystyle\sum_{j=1}^N \left[ S_{j}^{x}S_{j+1}^{x} + S_{j}^{y}S_{j+1}^{y} + \Delta S_{j}^{z}S_{j+1}^{z}\right] -2h\sum_{j=1}^{N} S_{j}^{z},
\]
with periodic boundary conditions, $S_{j+N}^{\alpha} = S_{j}^{\alpha}$, and $S_j^\alpha = \frac{1}{2}\sigma_j^\alpha$, where $\sigma_j^\alpha$ are Pauli matrices, and $\Delta$ is the uni-axial parameter strength, which is a ratio of $S^z$ interactions between $S^x$ or $S^y$ interactions. This model can interpolate continuously between classical Ising, quantum XXX, and quantum XY models. At $\Delta =0$, it turns to the quantum $XY$ or $XX0$ model which corresponds to free fermions on a lattice.  For $\Delta = 1$, ($\Delta = -1$), the anisotropic XXZ model Hamiltonian reduces to the isotropic (ferro)anti-ferromagnetic XXX model Hamiltonian. For $\Delta \to \pm \infty$, the model goes to an (ferro)anti-ferromagnetic Ising Model. 

The parameter $J$ defines an energy scale and only its sign is important: we observe a ferromagnetic ordering along the $x-y$ plane for positive values of $J$, and, for negative ones, we notice the anti-ferromagnetic alignment. The uni-axial parameter strength $\Delta$ distinguishes a planar regime $x-y$ (when $|\Delta| < 1$), from the axial alignment, (for $|\Delta| > 1$) \cite{BAnotes}. Thereby, it is useful to define two regimes: for $|\Delta| > 1$, the {\it Ising-like} regime and $|\Delta| < 1$, the {\it XY-like regime} in order to model materials possessing respectively an easy-axis and easy-plane magnetic anisotropies \cite{sariyer18}. 

Here we are interested in quantum correlations between the nearest and next to nearest neighbour spins in the XXZ spin-$1/2$ chain with $J=1$, at a temperature of $0K$, and zero external field ($h=0$). The matrix elements of $\varrho_{i+r}$ are written in function of expectation values which mean the correlation functions for nearest neighbour $r=1$, (for  $\varrho_{i+1}$), and the correlation functions for next-to-nearest neighbours $r=2$, (for $\varrho_{i+2}$), and they are given by a set of integral equations which can be found in \cite{biao.liang,Justino,Takahashi,Kato04,Kato03}. These two point correlation functions for the XXZ model at zero temperature and in the thermodynamic limit can be derived by using the Bethe Ansatz technique.
In Eq. \eqref{estado.biao}, due to the symmetry in the Hamiltonian model, it is presented the two-qubit reduced density matrix of sites $i$ and $i+r$, for $r=1,2$, in the thermodynamic limit, written in the basis $\ket{1} =\ket{\uparrow\uparrow}$, $\ket{2} = \ket{\uparrow\downarrow}$, $\ket{3} = \ket{\downarrow\uparrow}$ and $\ket{4} = \ket{\downarrow\downarrow}$, where $\ket{\uparrow}$ and $\ket{\downarrow}$ are the eigenstates of the Pauli $z$-operator \cite{biao.liang}: 
\begin{equation}
\varrho_{i+r} = 
\begin{pmatrix} 
    \varrho_{11}      &        0      &        0     &   0 \\ 
    0     &   \varrho_{22}   &   \varrho_{23}  &   0 \\
    0     &   \varrho_{32}   &   \varrho_{33}  &   0 \\
    0     &        0      &        0     &  \varrho_{44} 
\end{pmatrix}, 
\label{estado.biao}
\end{equation}
where $\varrho_{11} = \frac{1 + \mean{\sigma_{i}^z\sigma_{i+r}^z}}{4}$, $\varrho_{23} = \frac{ \mean{\sigma_{i}^x\sigma_{i+r}^x}}{2}$, and $\varrho_{22} = \frac{1 - \mean{\sigma_{i}^z\sigma_{i+r}^z}}{4}$, with $\varrho_{11} = \varrho_{44}$, $\varrho_{23} = \varrho_{32}$ and $\varrho_{22} = \displaystyle\varrho_{33}$.

In Fig. \ref{CSD.XXZ} we show the QSCM, $\mathcal{C}(\varrho_{i+1})$, for nearest neighbour, in contrast with the von Neumann entropy $S(\varrho_{i+1})$ and the trace distance $D(\varrho_{i+1},\mathcal{I})$ between $\varrho$ and the normalized identity matrix, all as a function of the uni-axial parameter strength $\Delta$. 
The XXZ model possess two critical points: the first-order transition occurs at $\Delta = -1$, and also a continuous phase transition shows up at $\Delta = 1$, \cite{Takahashi05}. An interesting feature of the QSCM  is the fact that it evince points of correlation transitions, related to the order-disorder transitions, which may not necessarily be connected with phase transition points. We take note of the cusp point in Fig. \ref{CSD.XXZ}, at $\Delta \approx 2.178$. For quantum statistical complexity measure correlations transitions appears to be interesting point in its evolution. 

\begin{figure}[H]
\centering
\includegraphics[width=0.9\linewidth]{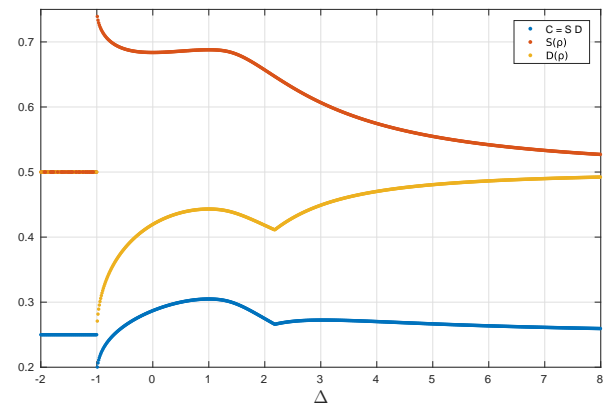}
\caption{ Quantum Statistical Complexity Measure $\mathcal{C}(\varrho_{i+1})$ (blue), von Neumann Entropy $S(\varrho_{i+1})$ (orange), and the disequilibrium function given by the Trace distance $D(\varrho_{i+1},\mathcal{I})$ (yellow) in function of $\Delta$. All measures were calculated for the two-qubit reduced density matrix of sites $i$ and $i+1$, $\varrho_{i+1}$, for $\Delta\in [-1,8]$, in the thermodynamic limit.}
\label{CSD.XXZ}
\end{figure}

In order to investigate the cusp point of $\mathcal{C}(\varrho_{i+1})$ at $\Delta = 2.178$, let us consider what happens with the state $\varrho_{i+1}$, given by Eq. \eqref{estado.biao}, as $\Delta$ varies. The state $\varrho_{i+1}$ given in Eq. \eqref{estado.biao} can be easily diagonalized, thus let us study the following matrix $\varrho_{i+1} - \mathbb{I}/4$,  which plays an important role in the quantum statistical complexity measure as already discussed. This matrix have the following eigenvalues: $\{\frac{1}{4} (2 \mean{\sigma_{i}^x\sigma_{i+1}^x}-\mean{\sigma_{i}^z\sigma_{i+1}^z}),\frac{1}{4} (-2\mean{\sigma_{i}^x\sigma_{i+1}^x}-\mean{\sigma_{i}^z\sigma_{i+1}^z}), \frac{1}{4}\mean{\sigma_{i}^z\sigma_{i+1}^z},\frac{1}{4}\mean{\sigma_{i}^z\sigma_{i+1}^z}\}$. 
As the $\Delta$ value increases in the interval $[1,3]$, correlation values in the $x$ direction also increase while correlations in $z$ decrease reaching the local minimum observed in Fig. \ref{CSD.XXZ}. In this interval, the eigenvalue $2\mean{\sigma_{i}^x\sigma_{i+1}^x} - \mean{\sigma_{i}^z\sigma_{i+1}^z}$ goes through zero, and this, therefore, should cause the correlation transition. This correlation transition is due the fact that this eigenvalue vanishes for some $\Delta$ in this interval, which should imply a change of orientation of spin correlations. 

By following this reasoning, in order to determine such points at which changes of orientation of spin correlations occurs, it is necessary to solve numerically some integral equations, given by the eigenvalues of Eq. \eqref{estado.biao}, which are functions of expected values given in Refs. \cite{biao.liang,Justino,Takahashi,Kato04,Kato03}.

This procedure has the objective of determining the solution for which values of $\Delta$ the following integral equation holds: $2 \mean{\sigma_{i}^x\sigma_{i+1}^x} - \mean{\sigma_{i}^z\sigma_{i+1}^z} = 0$. Due to the fact that the values of $\mean{\sigma_{i}^y\sigma_{i+1}^y} = \mean{\sigma_{i}^x\sigma_{i+1}^x}$, for this Hamiltonian, the solution of this equation indicates the point where planar $xy$-correlation decreases while $z$-correlations increases in absolute value, although we are already in the ferromagnetic phase. For $\Delta\to\infty$, the system moves towards a configuration that exhibits correlation only in the $z$-direction. Proceeding in the same way, by solving the other integral equation given in the eigenvalues set: $-2\mean{\sigma_{i}^x\sigma_{i+1}^x}-\mean{\sigma_{i}^z\sigma_{i+1}^z} =0$, we obtain a divergence solution for which $\Delta = -1$.

\begin{figure}[h]
\centering
\includegraphics[width=0.9\linewidth]{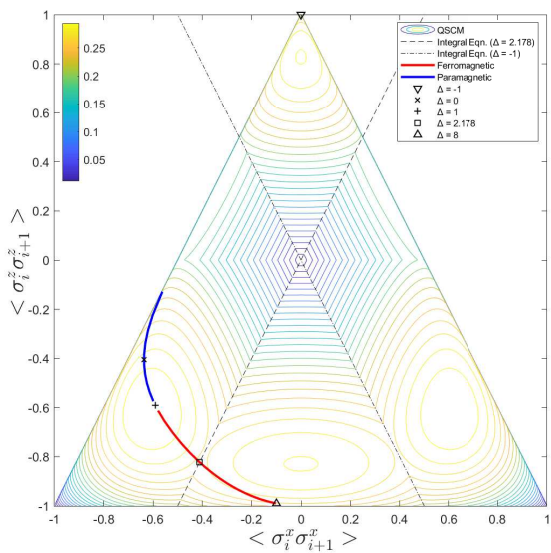}
\caption{Contour map of $\mathcal{C}(\varrho_{i+1})$, (QSCM), in function of the correlation functions $\mean{\sigma_{i}^x\sigma_{i+1}^x}$ and $\mean{\sigma_{i}^z\sigma_{i+1}^z}$. The dash inclined straight line represents the integral equation whose solution is $\Delta = 2.178$, and the dash-point straight line represents the curve for $\Delta = -1$, for which there is a divergence point. The indicated path inside the contour map shows the curve performed by the variation of $\mathcal{C}(\varrho_{i+1})$, inside the positive semi-definite density matrix space, for $\Delta\in[-1,8]$. The highlighted points are: $\Delta = -1$, ($\triangledown$); $\Delta = 0$, ($\times$); $\Delta = 1$, ($+$); $\Delta = 2.178$, ($\square$) and $\Delta = 8$, ($\vartriangle$). Also the ferromagnetic region (red) and the paramagnetic region (blue) are also represented in this path.}
\label{espaco.XXZ}
\end{figure}

In Fig. \ref{espaco.XXZ}, we call attention to a contour map of $\mathcal{C}(\varrho_{i+1})$ in function of  $\mean{\sigma_{i}^x\sigma_{i+1}^x}$, and $\mean{\sigma_{i}^z\sigma_{i+1}^z}$. The triangle region  represents the convex hull of positive semi-definite density matrices. The vertices of this triangle are given by: $(\mean{\sigma_{i}^x\sigma_{i+1}^x},\mean{\sigma_{i}^z\sigma_{i+1}^z}) = \{(-1,-1);\,(0,1);$ and $(1,-1)\}$. Along with the contour map of QSCM as a function of the correlation functions in $x$ and $z$ directions, the integral equations obtained while the two eigenvalues of $\varrho_{i+1} - \mathbb{I}/4$ goes to zero are also represented in Fig. \ref{espaco.XXZ}. These integral equations are represented by the two inclined straight lines (the dash and dash-dot ones). The dash and inclined straight line describes the integral equation whose solution is $\Delta = 2.178$. The dash-dot straight line represents the curve for $\Delta = -1$ solution, for which there exists a divergence point (the phase transition point). As previously mentioned, QSCM showed to be sensitive to correlation transitions. In Fig. \ref{espaco.XXZ}, the thick and colourful curve inside the contour map shows the path taken by $\mathcal{C}(\varrho_{i+1})$, while the values of correlations in $x$ and $z$ vary when $\Delta$ increases monotonically in the interval $[-1,8]$. This same path was also presented in Fig. \ref{CSD.XXZ}, on the blue curve. The blue part of the thick curve represents values for the correlations in which we have a paramagnetic state and the red part of the thick and colourful curve indicates the values for the ferromagnetic arrangement. Also we have highlighted some interesting points in this colourful curve by a $\times$: For $\Delta = -1$, ($\triangledown$); for $\Delta = 0$, ($\times$); $\Delta = 1$, ($+$); $\Delta = 2.178$, ($\square$) and for $\Delta = 8$, ($\vartriangle$). 

The Fig. \ref{complex2.XXZ} express QSCM for nearest neighbour, given by $\mathcal{C}(\varrho_{i+1})$, (blue), and for next-to-nearest neighbours, written as $\mathcal{C}(\varrho_{i+2})$, (orange), both in the thermodynamic limit in function of the uni-axial parameter strength $\Delta$.
\begin{figure}[h]
\centering
\includegraphics[width=0.9\linewidth]{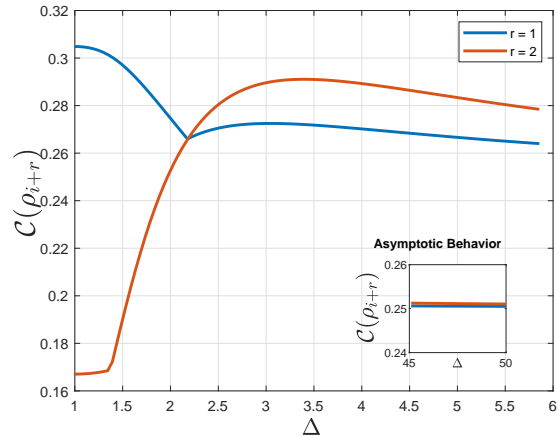}
\caption{$\mathcal{C}(\varrho_{i+r})$ for the two-qubit reduced density matrix of sites $i$ and $i+r$, for $r=1$ (blue), and for $r=2$ (orange), both in the thermodynamic limit in function of the uni-axial parameter strength $\Delta$. The sub-figure shows the asymptotic behaviour for large $\Delta$ for both cases. $\mathcal{C}(\varrho_{i+1})$ and $\mathcal{C}(\varrho_{i+2})\to 1/4$ when $\Delta\to\infty$.}
\label{complex2.XXZ}
\end{figure}
The asymptotic limit for both measures ($r=1,2$) is also presented in the sub-figure. As $\Delta\to\infty$, the behaviour of these two correlation functions $\mean{\sigma_{i}^x\sigma_{i+1}^x}\to 0$ and $\mean{\sigma_{i}^z\sigma_{i+1}^z}\to -1$, for both cases. In this limit, the density matrix of the system can be written as $\varrho_{i+r}\to diag\{0, 1/2, 1/2, 0\}$, for both cases and thus, $S(\varrho_{i+r})\to 1/2$. Also, $D(\varrho_{i+r},\mathcal{I})\to 1/2$, which make both measures, (for $r=1$ and for $r=2$), $\mathcal{C}(\varrho_{i+r})\to 1/4$. 
It is interesting to notice that $\mathcal{C}(\varrho_{i+1}) = \mathcal{C}(\varrho_{i+2})$ exactly at $\Delta = 2.178$. The QSCM for nearest neighbour is greater than the QSCM for next-to-nearest neighbour, \emph{i.e.;} $\mathcal{C}(\varrho_{i+1}) > \mathcal{C}(\varrho_{i+2})$, for $-1 \leq \Delta < 2.178$. For $\Delta > 2.178$, $\mathcal{C}(\varrho_{i+1}) < \mathcal{C}(\varrho_{i+2})$, until both goes to $1/4$, for large values of $\Delta$. 


\subsubsection{Study of the Derivatives of QSCM in  XXZ-\textonehalf\hspace{2.0mm}Model}
\label{phase.xxz}

In this section, we approach QSCM into the study of quantum phase transitions in Hamiltonian models. As phase transitions are phenomena in which physical quantities present discontinuities, we study the measure and its derivatives with respect to the parameter of interest in order to detect some discontinuity. The Heisenberg XXZ spin-$1/2$ model has two critical points for $h=0$: the first-order transition occurs at $\Delta = -1$, and a continuous phase transition shows up at $\Delta = 1$ \cite{Takahashi05}. 

The first-order transition point is easily recognisable in Fig. \eqref{CSD.XXZ}, because the QSCM express a discontinuity characteristic of a first order transition, for $\Delta = -1$. this property could also be obtained by solving the following integral equation: $-2\mean{\sigma_{i}^x\sigma_{i+1}^x}-\mean{\sigma_{i}^z\sigma_{i+1}^z} = 0$.

In Fig. \ref{D2complex.XXZ} we show the derivatives of QSCM up to third order with respect to the uni-axial parameter strength $\Delta$. Its form  could indicate a continuous phase transition of high derivative order for $\Delta \sim 1$. 
\begin{figure}[h]
\centering
\includegraphics[width=0.48 \linewidth]{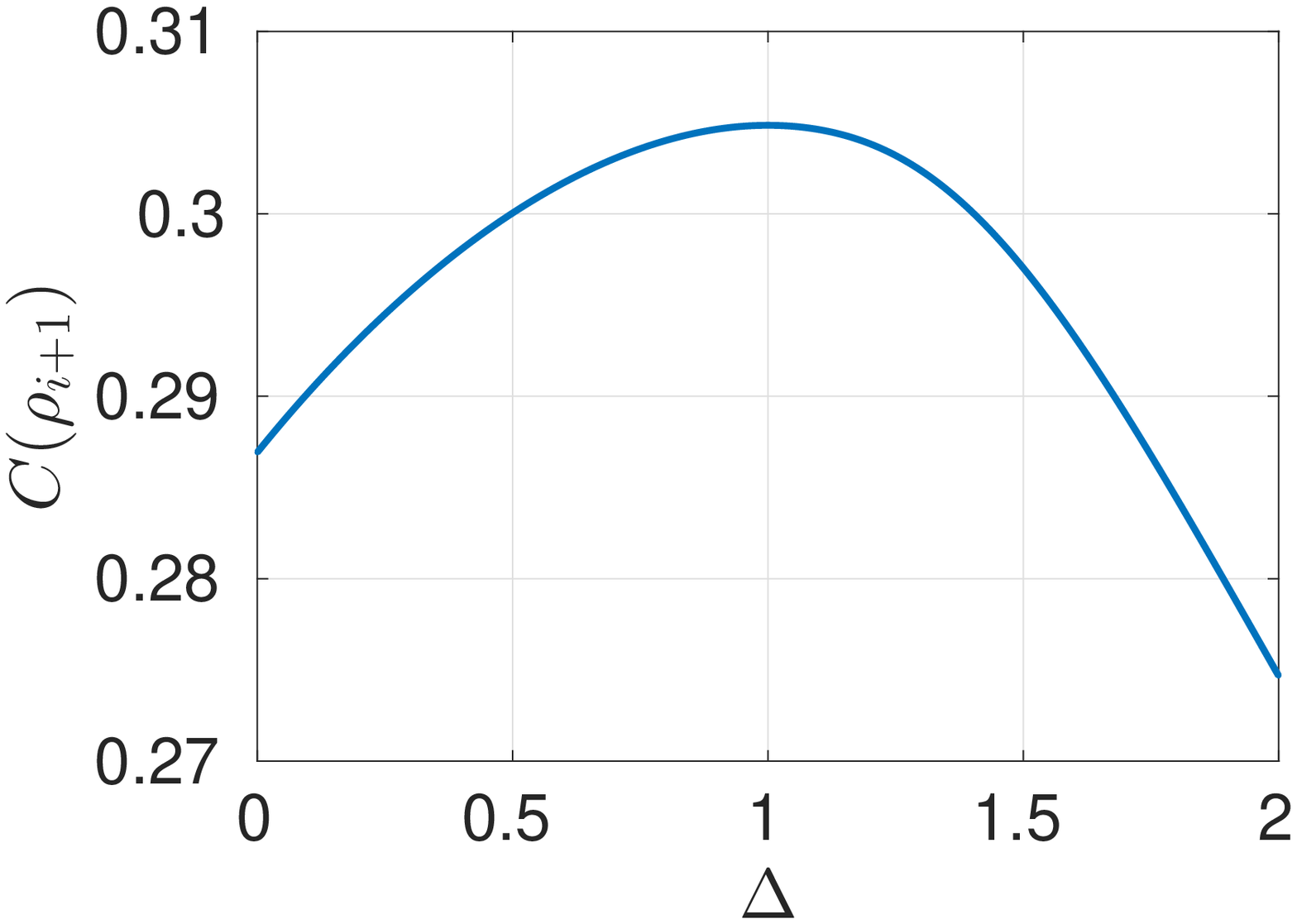}
\includegraphics[width=0.48 \linewidth]{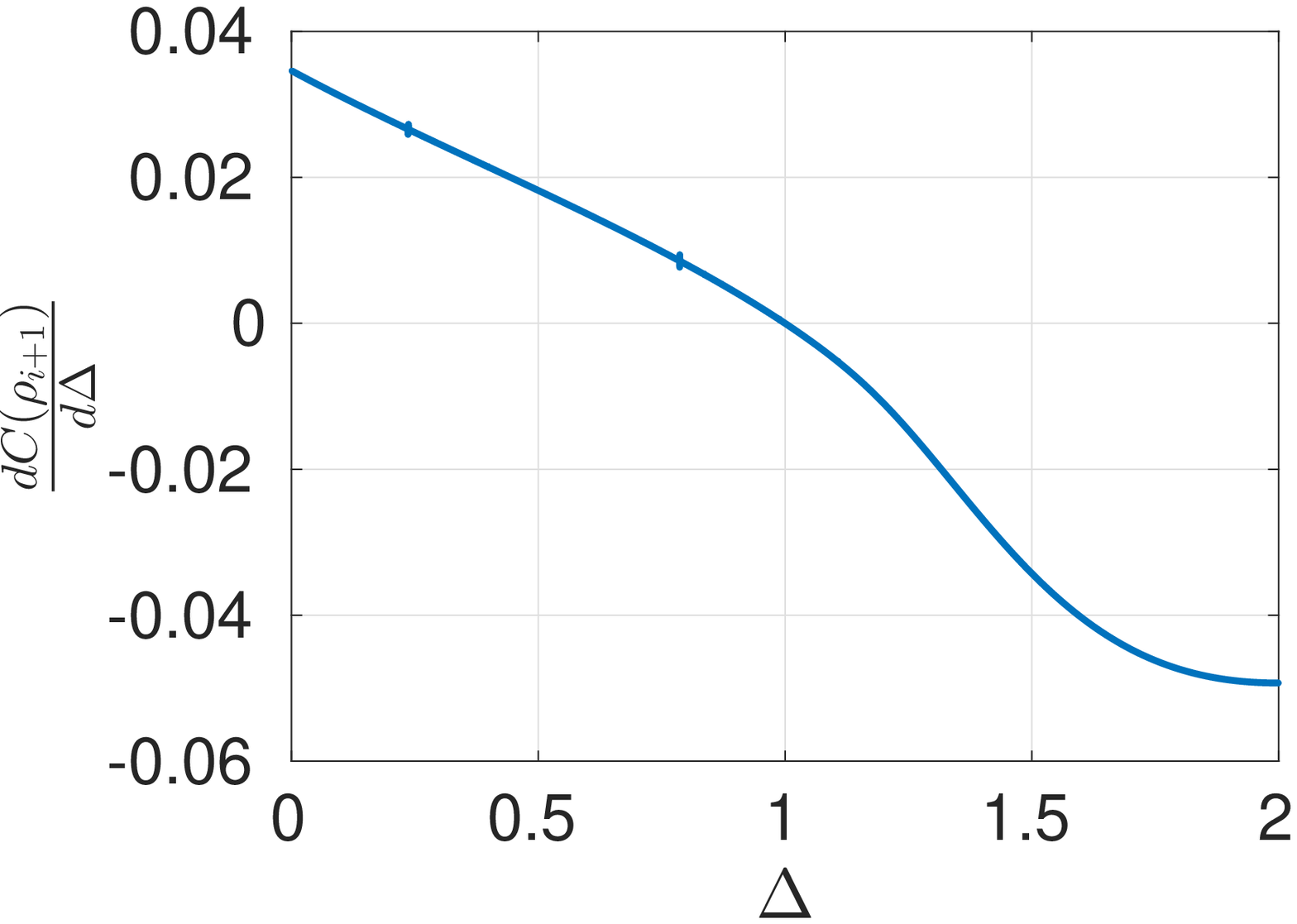}
\includegraphics[width=0.48 \linewidth]{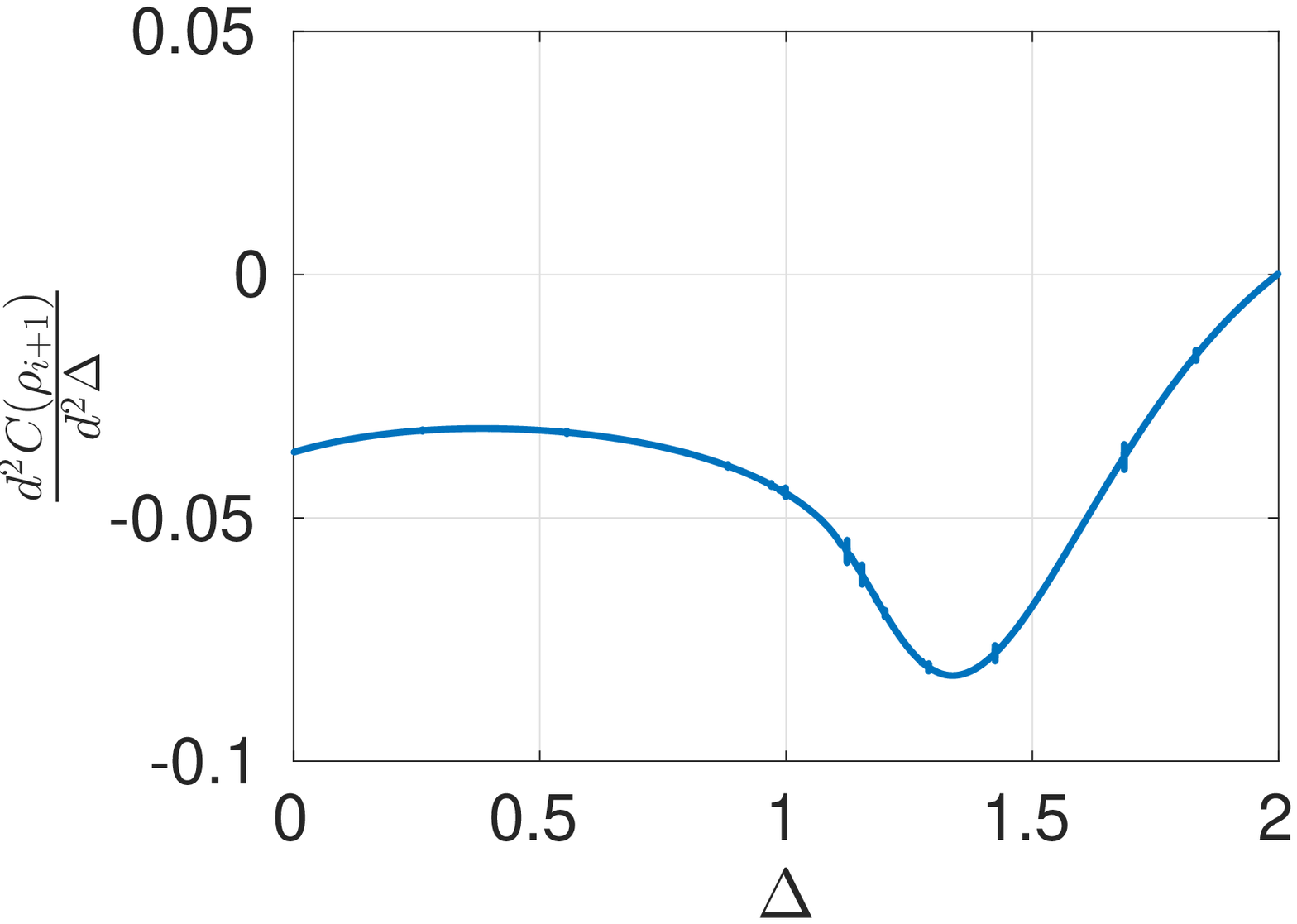}
\includegraphics[width=0.48 \linewidth]{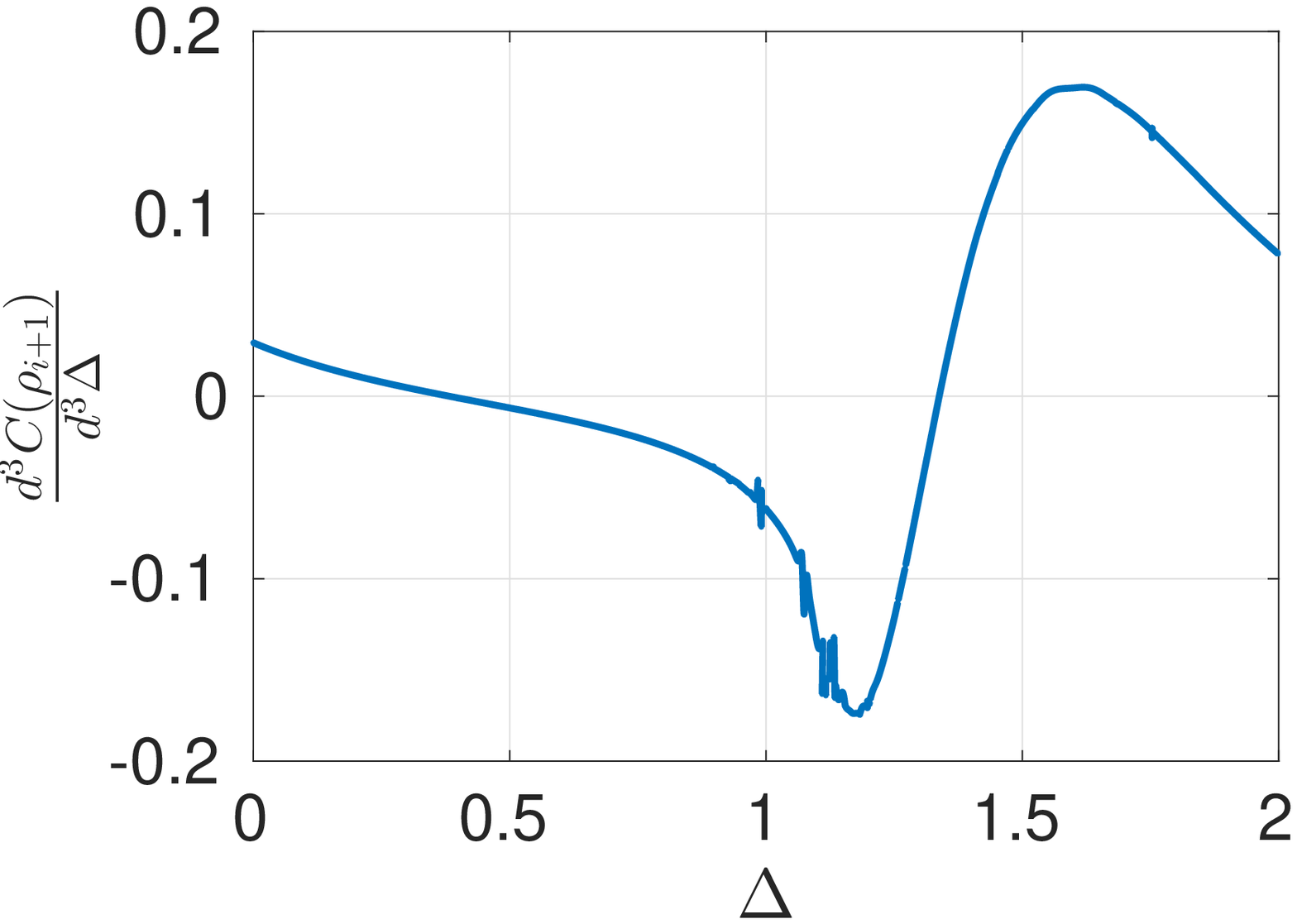}
\caption{Derivatives of the Quantum Statistical Complexity Measure with respect to $\Delta$: 
(top-left)   $\mathcal{C}(\varrho_{i+1})$,
 (top-right) $\frac{d \mathcal{C}(\varrho_{i+1})}{d\Delta}$
 (bottom-left) $\frac{d^2 \mathcal{C}(\varrho_{i+1})}{d\Delta^2}$
 (bottom-right) $\frac{d^3 \mathcal{C}(\varrho_{i+1})}{d\Delta^3}$.
 The small, but abrupt, oscillations of the derivatives around $\Delta = 1$ are not physical, but simply due to the numerical accuracy.
}
\label{D2complex.XXZ}
\end{figure}

One can notice that a discontinuity tends to converge to $\Delta \sim  1$ as we take into account higher orders of the derivative. This a evidence that quantum statistical measure of complexity could witness high order quantum phase transitions.

\section{Conclusions}
\label{conclusao}

We introduced a quantum version for the statistical complexity measure: the Quantum Statistical Measure of complexity (QSCM), and displayed some of its properties.
The measure has demonstrated to be useful and physically meaningful. It showed to possess several of the expected properties for a \emph{bona fide} complexity measure and demonstrated its possible usefulness in others areas of quantum information theory. 

We presented in the manuscript two applications of the QSCM, investigating the physics of two exactly solvable quantum Hamiltonian models, namely: the $1D$-quantum Ising Model and the Heisenberg XXZ spin-$1/2$ chain, both in the thermodynamic limit. 
Firstly we calculated the QSCM for one-qubit, in the Bloch's base, and we determined this measure as a function of the magnitude of the Bloch vector $r$. We computed this magnitude, in the thermodynamic limit, first by calculating analytically the measure for the one-qubit state reduced density matrix from $N$ spins, and later, in order to study the quantum phase transition for the $1D$-quantum Ising model, we performed the limit $N\to\infty$.

For the $1D$-quantum Ising model, we obtained the quantum phase transition point at $g=1$. In this way, we have found that the QSCM can be used as a signalling of quantum phase transitions for this model.

Secondly, we studied the Heisenberg XXZ spin-$1/2$ chain, and by means of QSCM we evince a point at which a correlation transition occurs for this model. Physically, at $\Delta = 2.178$, the planar $xy$-correlation decreases while the $z$-correlations decreases in value and increases in absolute value, reaching a minimum point, although we are already in the ferromagnetic organisation. This competition between these two different alignment of correlations demonstrates an order-disorder transition in which the measure was shown to be sensitive.    

We have studied the derivatives of the QSCM and they demonstrated to be sensitive to the quantum transition points. As a summary of this study for the Heisenberg XXZ spin-$1/2$, we can list: (\emph{i}) the quantum statistical measure of complexity characterize the first-order quantum phase transition at $\Delta = -1$, (\emph{ii}) also evince the  continuous quantum phase transition at $\Delta = 1$, and (\emph{iii}) a witness order-disorder transition at $\Delta = 2.178$, related to the alignment of the spins correlations.
\begin{acknowledgments}
This work was partially supported by  Brazilian agencies Fapemig, Capes, CNPq and  INCT-IQ through the project (465469/2014-0). T. D. also acknowledge the 
support from the Austrian Science Fund (FWF) through the project P 31339-N27.  F.I. acknowledge the financial support of the Brazilian funding agencies CNPQ (308205/2019-7) and FAPERJ.
\end{acknowledgments}


\clearpage

\hypertarget{sec:appendix}
\onecolumngrid
\appendix
\renewcommand\appendixname{Appendix}
\renewcommand\proof{\emph{Proof}.\ }
\renewcommand{\thesubsection}{A.\Roman{section}.\alph{subsection}}
\renewcommand{\thesection}{A.\Roman{section}}
\setcounter{equation}{0}
\numberwithin{equation}{section}
\setcounter{figure}{0}
\renewcommand{\thefigure}{A.\arabic{figure}}
\setcounter{table}{0}
\renewcommand{\thetable}{A.\arabic{table}}

\section*{Appendix}
\section{Sub-additivity over copies}
\begin{prop}
\label{appendixA}
Given a product state $\rho^{\otimes n}\in\mathcal{D}(\mathcal{H}^{\otimes n})$, QSCM is a sub-additive function:
 \begin{equation}
\mathcal{C}(\rho^{\otimes n})\leq n \mathcal{C}(\rho). 
\end{equation}
 Equality holds if we choose the disequilibrium function to be the quantum relative entropy, $S(\rho ||\mathcal{I}) = \log N - S(\rho)$, which is additive under tensor product. In this case, the measure will be additive $\mathcal{C}(\rho^{\otimes n}) = n \mathcal{C}(\rho)$.   
\end{prop}
\begin{proof}
von Neumann entropy is additive for product states: $S(\rho^{\otimes n}) = nS(\rho)$. This means that information contained in an uncorrelated system $\rho^{\otimes n}$, is equal to the sum of its constituents parts. This equality also holds for Shannon Entropy. However, the trace distance is sub-additive under respect to the tensor product: $D(\rho\otimes\rho ',\mathcal{I}\otimes \mathcal{I}) \leq D(\rho,\mathcal{I}) + D(\rho ',\mathcal{I}),\,\forall\,\rho,\rho '$ \cite{wilde17}. 

We will prove this proposition by induction and also we will only consider states with same dimension, \emph{i.e.;} $\text{dim}(\rho^{\otimes n}) = \text{dim}(\mathcal{I}^{\otimes n}),\forall\, n$. It is easy to observe that the proposition is true for $n=1$.  Let us suppose now, as an induction step, that for some arbitrary $n = k > 1,\,\,k\in\mathbb{N}$, $\mathcal{C}(\rho^{\otimes k}) \leq k \mathcal{C}(\rho)$. 
 \begin{align}
\mathcal{C}(\rho^{\otimes k+1}) &= \displaystyle\frac{S(\rho^{\otimes k}\otimes \rho)}{\text{log}(N^{k+1})}D(\rho^{\otimes k}\otimes \rho,\mathcal{I}^{\otimes k+1}),\\
\mathcal{C}(\rho^{\otimes k+1}) &= (k+1)\displaystyle\frac{S(\rho)}{\text{log}(N^{k+1})}D(\rho^{\otimes k}\otimes \rho,\mathcal{I}^{\otimes k+1})\\
\mathcal{C}(\rho^{\otimes k+1}) &= \frac{S(\rho)}{\text{log}(N)}D(\rho^{\otimes k}\otimes \rho,\mathcal{I}^{\otimes k+1}).
\end{align}
Using the sub-additivity property for the trace distance: $D(\rho^{\otimes k},\mathcal{I})\leq k D(\rho,\mathcal{I})$, and also using $S(\rho^{\otimes k}) = k S(\rho)$, for an arbitrary $n = k+1$
\begin{align}
    \mathcal{C}(\rho^{\otimes k+1}) &\leq(k+1)\frac{S(\rho)}{\text{log}(N)}D(\rho,\mathcal{I}),\\
    & = (k+1) \mathcal{C}(\rho)
\end{align}
Therefore $\mathcal{C}(\rho^{\otimes n}) \leq n \mathcal{C}(\rho),\,\,\forall\, n.$
\end{proof}

\end{document}